\documentclass[10pt,conference]{IEEEtran}
\usepackage[cmex10]{amsmath}
\usepackage{amssymb}
\usepackage{amsfonts}
\usepackage{bbold}
\usepackage{wrapfig}
\usepackage[hidelinks]{hyperref}
\usepackage{graphicx}
\usepackage{enumerate}
\usepackage{xspace}

\newcommand{\qp}{\ti{QuickPQE}\xspace}
\newcommand{\cs}{\mbox{$\mi{CmplSet}$}\xspace}

\newtheorem{example}{Example}
\newtheorem{definition}{Definition}
\newtheorem{proposition}{Proposition}
\newenvironment{proof}{\hspace{8pt}\ti{Proof:}}{}

\newtheorem{remark}{Remark}

\newcommand{\bm}[1]{{\mbox{\boldmath $#1$}}}

\newcommand{\di}[1]{\mbox{$\mi{Diam}(#1)$}\xspace}
\newcommand{\rch}[1]{\mbox{$\mi{Rch}(#1)$}\xspace}

\newcommand{\Cls}[2]{\mbox{$#1^{\vec{#2}'}$}}
\newcommand{\imp}{\Rightarrow}

\newcommand{\pnt}[1]{{\mbox{$\vec{#1}$}}}
\newcommand{\Pnt}[1]{{\mbox{$\vec{#1}\,'$}}}
\newcommand{\ppnt}[2]{{\mbox{$\vec{#1}_{#2}$}}}

\newcommand{\V}[1]{\mbox{$\mathit{Vars}(#1)$}}

\newcommand{\s}[1]{\mbox{$\{#1\}$}}

\newcommand{\nGz}[2]{$G_{non-\{z\}}$}

\newcommand{\prr}[1]{\mi{Prev}(\boldsymbol{q})}

\newcommand{\mi}[1]{\mathit{#1}}
\newcommand{\ti}[1]{\textit{#1}}
\newcommand{\tb}[1]{\textbf{#1}}

\newcommand{\Tt}{\>\>}

\newcommand{\prob}[2]{\mbox{$\exists{#1} [#2]$}}
\newcommand{\Prob}[3]{\mbox{$\exists{#1}\exists{#2}[#3]$}}

\newcommand{\Comment}[1]{}

\newcommand{\aps}[1]{\mbox{$\mathbb{#1}$}}
\newcommand{\abs}[1]{\mbox{$\mathcal{#1}$}}
\newcommand{\Abs}[2]{\mbox{$\mathcal{#1}^{\mi{#2}}$}}

\begin{document}

\title{Generation Of A Complete Set Of Properties}

\author{\IEEEauthorblockN{Eugene Goldberg} 
\IEEEauthorblockA{
eu.goldberg@gmail.com}}

\maketitle

\begin{abstract}
One of the problems of formal verification is that it is not
functionally complete due the incompleteness of specifications.  An
implementation meeting an incomplete specification may still have a
lot of bugs. In testing, where achieving functional completeness is
infeasible it is replaced with \ti{structural} completeness. The
latter implies generation of a set of tests probing every piece of a
design implementation. We show that a similar approach can be used in
formal verification. The idea here is to generate a property of the
implementation at hand that is not implied by the specification.
Finding such a property means that the specification is not
complete. If this is an \ti{unwanted} property, the implementation is
buggy. Otherwise, a new specification property needs to be
added. Generation of implementation properties related to different
parts of the design followed by adding new specification properties
produces a \ti{structurally-complete specification}. Implementation
properties are built by \ti{partial quantifier elimination}, a
technique where only a part of the formula is taken out of the scope
of quantifiers. An implementation property is generated by applying
partial quantifier elimination to a formula defining the ``truth
table'' of the implementation. We show how our approach works on
specifications of combinational and sequential circuits.
\end{abstract}

\section{Introduction}
One of the problems of formal verification is that it is functionally
incomplete. Let us consider this problem by the example of a
combinational design. Suppose a set \abs{P} =
\s{P_1(X,Z),\dots,P_k(X,Z)} of formulas\footnote{In this paper, we consider only propositional formulas. We assume that
every formula is in conjunctive-normal form (CNF).  A \ti{clause} is a
disjunction of literals (where a literal of a Boolean variable $w$ is
either $w$ itself or its negation $\overline{w}$). So a CNF formula
$H$ is a conjunction of clauses: $C_1 \wedge \dots \wedge C_k$. We
also consider $H$ as the \ti{set of clauses} \s{C_1,\dots,C_k}.
}
specify properties of a combinational circuit to be designed. Here $X$
and $Z$ are sets of input and output variables of this circuit
respectively\footnote{For the sake of simplicity, in the introduction, we assume that
properties $P_i(X,Z)$ depend on \ti{all} input/output variables.  In
Section~\ref{sec:gen_iprops_comb}, we consider a more general case
where a property depends on a \ti{subset} of $X \cup Z$.
}. (A correct implementation
has to exclude the input/output behaviors \ti{falsifying} $P_i$,
$i\!=\!1,\dots,\!k$.)  Let $N(X,Y,Z)$ be a circuit implementing the
specification \abs{P} above where $Y$ is the set of internal
variables. Let $F(X,Y,Z)$ be a formula describing the functionality of
$N$. That is every consistent assignment to the variables of $N$
corresponds to a satisfying assignment of $F$ and vice versa (see
Section~\ref{sec:basic}). The circuit $N$ satisfies property $P_i$, $
1 \le i \le k$ iff \mbox{$F\!\imp\!P_i$}. The circuit $N$ meets the
specification \abs{P} iff $F \imp (P_1 \wedge \dots \wedge P_k)$.

Unfortunately, the fact that $N$ satisfies its specification does not
mean that the former is correct. (For instance, if \abs{P} consists
only of one property $P$ where $P \equiv 1$, \ti{any} circuit meets
\abs{P}.) One also needs to check if \abs{P} is \ti{complete}.  This
comes down to checking if $P_1 \wedge \dots \wedge P_k \imp
\prob{Y}{F}$.  Here \prob{Y}{F} specifies the truth table of $N$. If
this implication does not hold, some input/output behaviors of $N$ are
not defined by \abs{P} i.e. the latter is incomplete. Note that
checking the completeness of \abs{P} is inherently hard because it
requires some form of quantifier elimination (\tb{QE}) for
\prob{Y}{F}.

In testing, the incompleteness of functional verification is addressed
by using a set of tests that is complete \ti{structurally} rather than
functionally.  Structural completeness is achieved by probing every
piece of the design under test. In this paper, we show that a similar
approach can be applied to formal verification. This approach is based
on two ideas. The first idea is to check the completeness of the
specification \abs{P} by generating \ti{implementation} properties
i.e. those satisfied by $N$.  Let $Q(X,Z)$ be a property of $N$ (and
so $F \imp Q$).  If \mbox{$P_1 \wedge \dots \wedge P_k \not\imp Q$},
then the specification \abs{P} is \ti{incomplete}. If $Q$ is an
unwanted property, $N$ is buggy (and it should be modified so that it
does not satisfy $Q$).  Otherwise, a new property should be added to
the specification \abs{P} to make the latter imply $Q$. A trivial way
to achieve this goal is just to add to \abs{P} the property $Q$
itself.

The second idea is to generate implementation properties by a
technique called \tb{partial} QE
(\tb{PQE})~\cite{cert_tech_rep,hvc-14}. In terms of formula
\prob{Y}{F}, PQE takes a subset of clauses of $F$ out of the scope of
quantifiers.  (So QE is special case of PQE where the entire formula
is taken out of the scope of quantifiers.) This results in generation
of a formula $Q(X,Z)$ implied by $F$ i.e.  a property of
$N$. Importantly, by taking different subsets of clauses of $F$ out of
the scope of quantifiers, one builds a \ti{structurally complete set
  of properties}. By updating specification properties every time an
implementation property proves \abs{P} incomplete, one gets a
structurally complete specification. Importantly, by using clause
splitting and varying the size of the subformula taken out of the
scope of quantifiers one can control the complexity of PQE and hence
that of property generation. The latter ranges from essentially
\ti{linear} (for properties specifying the input/output behavior of
$N$ for a single test) to exponential.

Incompleteness of the specification \abs{P} may lead to two kinds of
bugs. A bug of the first kind that we mentioned above occurs when $N$
has an unwanted property. In this case, $N$ \ti{excludes} some
\ti{correct} input/output behaviors. (An example of an unwanted
property is given in Appendix~\ref{app:unwanted}.) A bug of the second
kind occurs when $N$ \ti{allows} some \ti{incorrect} input/output
behaviors. This type of bugs can be exposed by generating properties
that are \ti{inconsistent} with $N$. (As opposed to the implementation
properties that are \ti{consistent} with $N$ by definition.) Such
inconsistent properties are meant to imitate the missing properties of
\abs{P} that are not satisfied by $N$ (if any). Tests falsifying
inconsistent properties may expose incorrect input/output behaviors
allowed by $N$.  These properties can also be generated by
PQE. Besides, one can follow the same idea of structural completeness
by building a set of inconsistent properties relating to different
parts of $N$. However, this topic is beyond the scope of this
paper. (It is covered in~\cite{inc_props}.) So here, we consider
generation of a specification that is structurally complete only with
respect to \ti{consistent} properties of the implementation at hand.

The contribution of this paper is as follows. First, we describe
generation of implementation properties by PQE. Second, we show that
clause splitting allows to reduce the complexity of PQE (and hence the
complexity of property generation) to virtually linear. The latter
result also shows that PQE can be exponentially more efficient than
QE. Third, we sketch an algorithm for generation of a structurally
complete specification.

This paper is organized as follows. Basic definitions are given in
Section~\ref{sec:basic}. Section~\ref{sec:gen_iprops_comb} describes
generation of implementation properties of combinational circuits by
PQE. Generation of properties specifying the input/output behavior of
a single test is discussed in Section~\ref{sec:st_prop}.
Section~\ref{sec:cmpl_set} presents a procedure for making a
specification structurally complete. In Sections~\ref{sec:ext_to_seq}
and~\ref{sec:gen_iprops_seq} we extend our approach to sequential
circuits. Some concluding remarks are made in Section~\ref{sec:concl}.


\section{Basic Definitions}
\label{sec:basic}
%
%
\begin{definition}
Let $V$ be a set of variables. An \tb{assignment} \pnt{v} to $V$ is a
mapping $V'~\rightarrow \s{0,1}$ where $V' \subseteq V$.  We will
refer to \pnt{v} as a \tb{full assignment} to $V$ if $V' = V$. 
\end{definition}

From now on, by saying ``an assignment to a set of variables'' we mean
a \ti{full} assignment, unless otherwise stated.
%
%
\begin{definition}
  \label{def:vars}
Let $F$ be a formula. \bm{\V{F}} denotes the set of variables of $F$.
\end{definition}

%
%
\begin{definition}
  \label{def:qe_prob}
Let $H(W,V)$ be a formula where $W,V$ are disjoint sets of Boolean
variables.  The \tb{Quantifier Elimination (QE)} problem specified by
\prob{W}{H} is to find a formula $H^*(V)$ such that \bm{H^* \equiv
  \prob{W}{H}}.
\end{definition}

%
%
\begin{definition}
  \label{def:pqe_prob}
Let $H_1(W,V)$, $H_2(W,V)$ be Boolean formulas where $W,V$ are sets of
Boolean variables.  The \tb{Partial QE} (\tb{PQE}) problem is to find
a formula $H^*_1(V)$ such that \bm{\prob{W}{H_1 \wedge H_2} \equiv
  H^*_1 \wedge \prob{W}{H_2}}.  We will say that $H^*_1$ is obtained
by taking $H_1$ out of the scope of quantifiers in \prob{W}{H_1 \wedge
  H_2} Formula $H^*_1$ is called a \tb{solution} to PQE.
\end{definition}
%
%
\begin{remark}
  \label{rem:noise}
Note that if $H^*_1$ is a solution to the PQE problem above and a
clause $C \in H^*_1$ is implied by $H_2$ \ti{alone}, then $H^*_1
\setminus \s{C}$ is a solution too. If all clauses of $H^*_1$ are
implied by $H_2$, an empty set of clauses is a solution too. In this
case, $H^*_1 \equiv 1$ and $H_1$ is redundant in \prob{V}{H_1 \wedge
  H_2}.
\end{remark}

Let $N(X,Y,Z)$ be a combinational circuit where $X,Y,Z$ are sets of
input, internal and output variables respectively. We will say that a
formula $F(X,Y,Z)$ \tb{defines} $N$ if every consistent assignment to
the variables of $N$ corresponds to a satisfying assignment of $F$ and
vice versa~\cite{tseitin}. Let $N$ consist of gates
$g_1,\dots,g_k$. The formula $F$ can be built as $G_1 \wedge \dots
\wedge G_k$ where $G_i, 1 \leq i \leq k$ is a formula defining gate
$g_i$. Formula $G_i$ is constructed as a conjunction of clauses
falsified by the incorrect combinations of values assigned to
$G_i$. Then every assignment satisfying $G_i$ corresponds to a
consistent assignment of values to $g_i$ and vice versa.

\begin{example}
\label{exmp:circ_form}
Let $g$ be a 2-input AND gate specified by $v_3 = v_1 \wedge
v_2$. Then a formula $G$ defining $g$ is constructed as \mbox{$C_1
  \wedge C_2 \wedge C_3$} where $C_1 = v_1 \vee \overline{v}_3$, $C_2
= v_2 \vee \overline{v}_3$, $C_3 = \overline{v}_1 \vee \overline{v}_2
\vee v_3$. Here, the clause $C_1$, for instance, is falsified by the
assignment ($v_1=0,v_3=1$) that is inconsistent with the truth table
of $g$.
\end{example}

\section{Generation Of Implementation Properties}
\label{sec:gen_iprops_comb}
Let $N(X,Y,Z)$ be a combinational circuit where $X,Y,Z$ are sets of
input, internal and output variables respectively.  Let $F(X,Y,Z)$ be
a formula defining $N$.  Let $H$ be a non-empty subset of clauses of
$F$.  Consider the PQE problem of taking $H$ out of the scope of
quantifiers in \prob{W}{F} where $Y \subseteq W \subset \V{F}$. Let
formula $Q(V)$ be a solution to this problem i.e. $\prob{W}{F} \equiv
Q \wedge \prob{W}{F \setminus H}$. (Here $V$ denotes $\V{F} \setminus
W$ and so $V \subseteq (X \cup Z)$.) Since $Q$ is implied by $F$, it
is a \tb{property} of the circuit $N$. Note that by taking different
subsets of $F$ out of the scope of quantifiers in \prob{W}{F} one gets
different properties.

Intuitively, the smaller $H$, the easier taking $H$ out of the scope
of quantifiers.  So, the simplest case of the PQE problem above is
when a single clause of $F$ is taken out of the scope of
quantifiers. However, the complexity of PQE can be reduced much more
by using clause splitting to transform $F$.
%
%
\begin{definition}
\label{def:cls_split}
Let $R = \s{v_1,\dots,v_m}$ be a subset of \V{F}. Let
$l(v_1),\dots,l(v_m)$ be a set of literals.  Let $C$ be a clause of
$F$ such that $R \cap \V{C} = \emptyset$. The \tb{splitting} of $C$ on
variables of $R$ is to replace $C$ with clauses $C \vee
l(v_1)$,$\dots$, $C \vee l(v_m)$, $C \vee \overline{l(}v_1) \vee \dots
\vee \overline{l(}v_m)$.
\end{definition}

The idea here is to take the clause $C \vee \overline{l(}v_1) \vee
\dots \vee \overline{l(}v_m)$ out of the scope of quantifiers
\ti{instead} of $C$.  In the next section, we show that such
replacement can reduce the complexity of PQE to essentially
\tb{linear}. This also proves that PQE can be \tb{exponentially
  simpler} than QE.

%
\section{Generation Of Single-Test Properties}
\label{sec:st_prop}
In this section, we use clause splitting to make the following two
points. First, by using clause splitting and PQE one can generate very
weak properties e.g. properties specifying the input-output behavior
of a circuit for a single test.  Second, by using clause splitting one
can reduce the complexity of PQE and, hence, that of property
generation. (Of course, this complexity reduction is achieved at the
expense of the property strength.) In particular, for the single-test
properties mentioned above, this complexity reduces to essentially
\ti{linear}.
%
%
\subsection{A single-test property}
In this section, we continue using the notation of the previous
section.  In particular, we assume that a formula $F(X,Y,Z)$ defines a
combinational circuit $N(X,Y,Z)$ where $X,Y,Z$ are sets of input,
internal and output variables respectively.
%
%
\begin{definition}
\label{def:st_prop}
Let \Pnt{x} be a \tb{test} (i.e. an assignment to $X$).  Let \Pnt{z}
be the output assignment produced for \Pnt{x} by $N$. We will call
formula $Q(X,Z)$ a \tb{single-test property} of $N$ if
\begin{enumerate}
\item $Q(\pnt{x},\pnt{z}) = 1$ for every \pnt{x} different from
  \Pnt{x} regardless of the value of \pnt{z};
\item $Q(\Pnt{x},\Pnt{z}) = 1$;
\item $Q(\Pnt{x},\pnt{z}) = 0$ for at least one \pnt{z} different from
  \Pnt{z}.
\end{enumerate}
\end{definition}

Informally, $Q$ is a single-test property of $N$ if it (partially)
specifies the input/output behavior of $N$ for a single test
\Pnt{x}. Namely, $Q$ excludes (some) output assignments that are
\ti{not} produced for \Pnt{x} by $N$. In
Subsection~\ref{ssec:quick_pqe}, we describe a procedure called
\ti{QuickPQE} that generates a single-test property.

%
%
\subsection{The PQE problem we consider in this section}
\label{ssec:pqe_prob}
For the sake of simplicity, in our exposition, we use a particular
clause of the formula describing an AND gate of $N$. (However, we
explain how to extend this exposition to an arbitrary clause of the
formula describing an arbitrary gate of $N$).  Let $g$ be an AND gate
of a circuit $N$ whose functionality is described by $v_3 =
v_1\!\wedge\!v_2$ (see Example~\ref{exmp:circ_form}). Let clause $C
\in F$ be equal to $\overline{v}_1\!\vee\!\overline{v_2}\!\vee v_3$.
This clause forces assigning the output variable $v_3$ of $g$ to 1
when the input variables $v_1$ and $v_2$ of $g$ are assigned 1. (In
the general case, $C \in F$ is one of the clauses specifying a gate
$g$ of $N$. The clause $C$ has one variable specifying the output of
$g$. The remaining variables of $C$ correspond to the input variables
of $g$.)

Assume for the sake of simplicity that $\V{C} \cap X = \emptyset$.
Consider splitting $C$ on the variables of $X = \s{x_1,\dots,x_m}$.
That is $C$ is replaced in $F$ with $m+1$ clauses $C \vee
l(x_1)$,$\dots$, $C \vee l(x_m)$, $C \vee \overline{l(}x_1) \vee \dots
\vee \overline{l(}x_m)$. Denote the last clause as $C'$ $($i.e.  $C' =
C \vee \overline{l(}x_1) \vee \dots \vee \overline{l(}x_m))$. Let $F'$
denote $F \setminus \s{C'}$.  The PQE problem we solve in the next
subsection is to take $C'$ out of the scope of quantifiers in
\prob{Y}{C' \wedge F'}.

%
%
\subsection{QuickPQE procedure}
\label{ssec:quick_pqe}
Now we present a procedure called \qp that takes $C'$ out of the scope
of quantifiers in \prob{Y}{C' \wedge F'}. Since \qp solves only a
particular subset of instances of the PQE problem, it is
incomplete. Our intention here is just to show that this subset of
instances can be solved efficiently. One can easily incorporate \qp
into a complete PQE algorithm. This simply requires adding a procedure
for checking if the current instance of PQE satisfies the definition
of Subsection~\ref{ssec:pqe_prob} and if so, calling \qp.

Let \Pnt{x} denote the assignment to $X$ falsifying the literals
$\overline{l(}x_1),\dots, \overline{l(}x_m)$ of $C'$. \ti{QuickPQE}
starts with applying \Pnt{x} to $N$. Let $\pnt{z}\,'$ be the output
assignment produced by $N$ for $\pnt{x}\,'$. Suppose that $v_1$ and/or
$v_2$ are assigned 0 when computing $\pnt{z}\,'$.  (In the general
case, this means that the clause $C$ and hence the clause $C'$ is
satisfied by an assignment to an \ti{input} variable of the gate $g$.)
Then \ti{QuickPQE} declares $C'$ redundant claiming that $\prob{Y}{C'
  \wedge F'} \equiv \prob{Y}{F'}$.

If both $v_1$ and $v_2$ are assigned 1, then \ti{QuickPQE} performs
one more run. (In the general case, this means that the literals of
$C$ corresponding to the input variables of the gate $g$ are
falsified.) In this run, \ti{QuickPQE} also applies input $\pnt{x}\,'$
but modifies the operation of the gate $g$. Namely, $g$ produces the
output value 0 (instead of the value 1 implied by assignment $v_1=1,
v_2=1$). Note that in the second run, the clause $C'$ is falsified.
One can view the second run as applied to the circuit $N$ whose
functionality is modified by removing the clause $C'$. If the second
run produces the same output assignment $\pnt{z}\,'$, then
\ti{QuickPQE} again declares $C'$ redundant.  Now, suppose that $N$
outputs an assignment \pnt{z^*} different from $\pnt{z}\,'$. Then
\ti{QuickPQE} produces the solution $Q(X,Z)$ consisting of clauses
$\Cls{B}{x} \vee l(z_1)$,\dots,$\Cls{B}{x} \vee l(z_p)$ where
\begin{itemize}
\item \Cls{B}{x}= $\overline{l(}x_1) \vee \dots \vee
  \overline{l(}x_m)$ (i.e.  \Cls{B}{x} is the longest clause falsified
  by $\pnt{x}\,'$);
\item $z_1,\dots,z_p$ are the output variables of $N$ assigned
  \ti{differently} in $\pnt{z}\,'$ and \pnt{z^*};
\item $l(z_1),\dots,\l(z_p)$ are literals satisfied by $\pnt{z}\,'$
  (and falsified by \pnt{z^*}).
\end{itemize}
%
%
\begin{proposition}
  \label{prop:lin_time}
Let $\mi{Nlits}(F)$ denote the number of literals of $F$.  Let
\ti{QuickPQE} be applied to the PQE problem of taking $C'$ out of the
scope of quantifiers in \prob{Y}{C' \wedge F'} (described in
Subsection~\ref{ssec:pqe_prob}). Then \ti{QuickPQE} produces a correct
result and the complexity of \ti{QuickPQE} is $\abs{O}(\mi{Nlits}(F) +
|X|*|Z|)$.
\end{proposition}

The proofs of propositions are given in Appendix~\ref{app:proof}.

%
%
\begin{proposition}
  \label{prop:st_prop}
Let $C'$ be non-redundant in \prob{Y}{C' \wedge F'}. Then the formula
$Q(X,Z)$ generated by \ti{QuickPQE} is a single-test property of $N$.
\end{proposition}

\section{Producing A Structurally Complete Set Of Properties}
\label{sec:cmpl_set}

In this section, we give an example of a procedure called \cs that
generates a structurally complete specification. The pseudocode of \cs
is shown in Figure~\ref{fig:cmpl_set}. \cs accepts
\begin{itemize}
\item a specification \abs{P} (i.e. a set of properties $P_1,\!\dots,\!P_k$)
\item an ``informal'' specification \Abs{P}{inf} that is used to tell
  if a property of $N$ is unwanted
\item an implementation $F(X,Y,Z)$ defining a circuit $N$
\item the set of variables $V \subseteq (X \cup Z)$ on which
  implementation properties will depend on.
\end{itemize}
\cs returns an unwanted property of $N$ exposing a bug (if any) or a
structurally complete specification \abs{P}. The existence of an
informal specification \Abs{P}{inf} is based on the assumption that,
for every input, the designer is able to tell if the output produced
by $N$ is incorrect.

%
%
\setlength{\intextsep}{4pt}
\setlength{\textfloatsep}{10pt}
\begin{figure}[h]
\centering
\small
\parbox{0cm}{\begin{tabbing}
aaa\=bb\=cc\= dd\= \kill
$\cs(\abs{P},\Abs{P}{inf},F,V)$\{\\
\tb{\scriptsize{1}}\> $\mi{Cls}:=F$    \\
\tb{\scriptsize{2}}\> while $(\mi{Cls} \neq \emptyset)$ \{ \\
\tb{\scriptsize{3}}\Tt  $C := \mi{PickCls}(Cls)$ \\
\tb{\scriptsize{4}}\Tt  $\mi{Cls} := \mi{Cls} \setminus \s{C}$ \\
\tb{\scriptsize{5}}\Tt  $Q := \mi{PQE}(F,C,V)$ \\
\tb{\scriptsize{6}}\Tt  $\mi{Clean}(Q,F,C)$ \\
\tb{\scriptsize{7}}\Tt  if $(\mi{Impl}(\abs{P},Q))$ continue \\
\tb{\scriptsize{8}}\Tt  if $(\mi{Unwanted}(\Abs{P}{inf},Q))$ return($Q$,$\mi{nil}$) \\
\tb{\scriptsize{9}}\Tt  $P := \mi{SpecProp}(\abs{P},\Abs{P}{inf},F,Q)$ \\
\tb{\scriptsize{10}}\Tt  $\abs{P} := \abs{P} \cup \s{P}$\} \\
\tb{\scriptsize{11}}\> return($\mi{nil},\abs{P}$)\}  \\
\end{tabbing}}
\vspace{-10pt}
\caption{The \cs procedure}
\vspace{3pt}
\label{fig:cmpl_set}
\end{figure}

\cs starts with initializing a copy \ti{Cls} of formula $F$ (line
1). Then \cs runs a 'while' loop until \ti{Cls} is empty. \cs starts
an iteration of the loop by extracting a clause $C$ from \ti{Cls}
(lines 3-4). Then it builds an implementation property $Q(V)$ as a
solution to the PQE problem \prob{W}{C \wedge F'} where $F' = F
\setminus \s{C}$ and $W = \V{F} \setminus V$ (line 5). That is
$\prob{W}{C \wedge F'} \equiv Q \wedge \prob{W}{F'}$.  One can view
$Q$ as a property ``probing'' the part of $N$ represented by $C$. Then
\cs calls the procedure called \ti{Clean} (line 6) to remove the
clauses implied by $F'$ from $Q$ (see Remark~\ref{rem:noise}). At this
point, $Q$ consists only of clauses whose derivation involved the
clause $C$.

Then \cs checks if $P_1 \wedge \dots \wedge P_k \imp Q$ (line 7).  If
so, then a new iteration starts. Otherwise, the current specification
\abs{P} is incomplete, which requires modification of $N$ or
\abs{P}. If $Q$ is an unwanted property, \cs returns it as a proof
that $N$ is buggy (line 8). In this case, $Q$ excludes some correct
input/output behaviors. To decide if this is the case, the informal
specification \Abs{P}{inf} mentioned above is applied.  If $Q$ is a
desired property, \cs generates a new \ti{specification} property $P$
such that $P_1 \wedge \dots \wedge P_k \wedge P \imp Q$ and adds it to
\abs{P} (lines 9-10).  A trivial way to update \abs{P} is just to use
$Q$ as a new specification property (i.e. $P = Q$).  If \cs terminates
the loop without finding a bug, it returns \abs{P} as a structurally
complete specification.

\section{Extending Idea To Sequential Circuits}
\label{sec:ext_to_seq}
In this section and Section~\ref{sec:gen_iprops_seq}, we extend our
approach to sequential circuits. Subsections~\ref{ssec:seq_defs}
and~\ref{ssec:stutt} provide some definitions.
Subsection~\ref{ssec:hi_lvl} gives a high-level view of building a
structurally complete specification for a sequential circuit (in terms
of safety properties).
%
%
\subsection{Some definitions}
\label{ssec:seq_defs}
Let $M(S,X,Y,S')$ be a sequential circuit. Here $X,Y$ denote input and
internal combinational variables respectively and $S,S'$ denote the
present and next state variables respectively. (For the sake of
simplicity, we assume that $M$ does not have any combinational output
variables.) Let $F(S,X,Y,S')$ be a formula describing the circuit
$M$. The formula $F$ is built for $M$ in the same manner as for a
combinational circuit (see Section~\ref{sec:basic}).  Let $I(S)$ be a
formula specifying the \tb{initial states} of $M$.  Let $T(S,S')$
denote \Prob{X}{Y}{F} i.e. the \tb{transition relation} of $M$.

A \tb{state} \pnt{s} is an assignment to $S$. Any formula $P(S)$ is
called a \ti{safety property} for $M$. A state \pnt{s} is called a
$P$-\ti{state} if $P(\pnt{s})=1$. A state \pnt{s} is called
\ti{reachable} \ti{in} $n$ \ti{transitions} (or in $n$\ti{-th}
\ti{time frame}) if there is a sequence of states
\ppnt{s}{1},\dots,\ppnt{s}{n+1} such that \ppnt{s}{1} is an $I$-state,
$T(\ppnt{s}{i},\ppnt{s}{i+1})=1$ for $i=1,\dots,n$ and
\ppnt{s}{n+1}=\pnt{s}.

We will denote the \ti{reachability diameter} of $M$ with initial
states $I$ as \bm{\di{M,I}}. That is if $n=\di{M,I}$, every state of
$M$ is reachable from $I$-states in at most $n$ transitions. We will
denote as \bm{\rch{M,I,n}} a formula specifying the set of states of
$M$ reachable from $I$-states in $n$ transitions. We will denote as
\bm{\rch{M,I}} a formula specifying all states of $M$ reachable from
$I$-states.  A property $P$ \ti{holds} for $M$ with initial states
$I$, if no $\overline{P}$-state is reachable from an $I$-state.

%
%

\subsection{Stuttering}
\label{ssec:stutt}

In the following explanation, for the sake of simplicity, we assume
that the circuit $M$ above has the \tb{stuttering} feature. This means
that $T(\pnt{s},\pnt{s})$=1 for every state \pnt{s} and so $M$ can
stay in any given state arbitrarily long.  If $M$ does not have this
feature, one can introduce stuttering by adding a combinational input
variable $v$.  The modified circuit works as before if $v=1$ and
remains in its current state if $v=0$.

On one hand, introduction of stuttering does not affect the
reachability of states of $M$. On the other hand, stuttering
guarantees that the transition relation of $M$ has two nice
properties.  First, $\prob{S}{T(S,S')} \equiv 1$, since for every next
state \Pnt{s}, there is a ``stuttering transition'' from \pnt{s} to
\Pnt{s} where \pnt{s} = \Pnt{s}.  Second, if a state is unreachable in
$M$ in $n$ transitions it is also unreachable in $i$ transitions if $i
< n$. Conversely, if a state is reachable in $M$ in $n$ transitions,
it is also reachable in $i$ transitions where $i > n$.
\begin{remark}
\label{rem:stutt}
Note that for a circuit $M$ with the stuttering feature, formula
\rch{M,I,n} specifies not only the states reachable in $n$ transitions
but also those reachable in \ti{at most} $n$ transitions.
\end{remark}

%
%
\subsection{High-level view}
\label{ssec:hi_lvl}
In this paper, we consider a specification of the sequential circuit
$M$ above in terms of safety properties. So, when we say a
specification property $P(S)$ of $M$ we \ti{mean a safety property}.
Let $F_{1,i}$ denote $F_1 \wedge \dots \wedge F_i$ where $F_j$, $1
\leq j \leq i$ is the formula $F$ in $j$-th time frame i.e. expressed
in terms of sets of variables $S_j,X_j,Y_j,S_{j+1}$.  Formula
\rch{M,I,n} can be computed by QE on formula \prob{W_{1,n}}{I_1 \wedge
  F_{1,n}}.  Here $I_1 = I(S_1)$ and $W_{1,n} = \V{F_{1,n}} \setminus
S_{n+1}$. If $n \ge \di{M,I}$, then \rch{M,I,n} is also \rch{M,I}
specifying all states of $M$ reachable from $I$-states.

Let $\abs{P} = \s{P_1,\dots,P_k}$ be a set of properties forming a
\ti{specification} of a sequential circuit with initial states defined
by $I$.  Let a sequential circuit $M$ be an implementation of the
specification \abs{P}.  So, every property $P_i$,$i=1,\dots,k$ holds
for $M$ and $I$. Verifying the completeness of \abs{P} reduces to
checking if $P_1 \wedge \dots \wedge P_k \imp \rch{M,I}$. Assume that
computing \rch{M,I} is hard. So, one does not know if the
specification \abs{P} is complete. Then one can use the approach
described in the previous sections to form a specification that is
complete \ti{structurally} rather than functionally.

We exploit here the same idea of using PQE to compute properties of
$M$ i.e.  \ti{implementation} properties. Let $Q$ be such a property.
If $P_1 \wedge \dots \wedge P_k \not\imp Q$, then the specification
\abs{P} is \ti{incomplete}. If some states falsifying $Q$ (and hence
unreachable from $I$-states) should be reachable, $M$ is buggy and
must be modified. Otherwise, one needs to update \abs{P} by adding a
specification property $P$ to guarantee that $P_1 \wedge \dots \wedge
P_k \wedge P \imp Q$. The simplest way to achieve this goal is just to
add $Q$ to \abs{P}. Using a procedure similar to that shown in
Fig.~\ref{fig:cmpl_set} one can construct a structurally complete
specification.

\section{Generation Of Safety Properties}
\label{sec:gen_iprops_seq}
In this section, we continue using the notation of the previous
section. Here, we discuss generation of properties for a sequential
circuit $M(S,X,Y,S')$, i.e. \ti{implementation} properties.
Subsection~\ref{ssec:rch_diam_on} considers the case where the
reachability diameter \di{M,I} is known.  (In~\cite{mc_no_inv}, we
showed that one can use PQE to find \di{M,I} without generation of all
reachable states.) Subsection~\ref{ssec:rch_diam_off} describes an
approach to generation of properties when \di{M,I} is not known.
%
%
\subsection{The case of known reachability diameter}
\label{ssec:rch_diam_on}
As we mentioned in Subsection~\ref{ssec:hi_lvl}, formula \rch{M,I} can
be obtained by QE on \prob{W_{1,n}}{I_1 \wedge F_{1,n}} where $n \ge
\di{M,I}$.  Here $I_1 = I(S_1)$, $F_{1,n} = F_1 \wedge \dots \wedge
F_n$, and $W_{1,n} = \V{F_{1,n}} \setminus S_{n+1}$.

Below we show how one can build a property of $M$ by PQE. Let $C$ be a
clause of $F_{1,n}$. Let $Q(S_{n+1})$ be a solution to the PQE problem
of taking $C$ out of the scope of quantifiers in \prob{W_{1,n}}{I_1
  \wedge C \wedge F'_{1,n}} where $F'_{1,n} = F_{1,n} \setminus
\s{C}$. That is \prob{W_{1,n}}{I_1 \wedge C \wedge F'_{1,n}} $\equiv$
$Q \wedge$ \prob{W_{1,n}}{I_1 \wedge F'_{1,n}}. Let us show that $Q$
is a \ti{property} of $M$. Let \pnt{s} be a state falsifying $Q$
i.e. \pnt{s} in unreachable from an $I$-state in $n$ transitions.  On
one hand, since $M$ has the stuttering feature, \pnt{s} cannot be
reached in $i$ transitions where $i \leq n$. On the other hand, since
$n \ge \di{M,I}$, \pnt{s} cannot be reached in $i$ transitions where
$i > n$. So all states falsifying $Q$ are unreachable and thus $Q$ is
a property of $M$. By taking different clauses of $F_{1,n}$ out of the
scope of quantifiers one can generate different implementation
properties. Following a procedure similar to that of
Fig.~\ref{fig:cmpl_set}, one can generate a specification of $M$ that
is structurally complete.

%
%
\subsection{The case of unknown reachability diameter}
\label{ssec:rch_diam_off}
Suppose that the reachability diameter of $M$ is unknown. Then one
needs to modify the procedure of the previous subsection as
follows. Let $Q(S_{n+1})$ be a solution to the PQE problem of taking
$C$ out of the scope of quantifiers in \prob{W_{1,n}}{I_1 \wedge C
  \wedge F'_{1,n}} where $F'_{1,n} = F_{1,n} \setminus \s{C}$. Assume
that $n < \di{M,I}$. Then $Q$ is not a property of $M$. One can only
guarantee that the states falsifying $Q$ cannot be reached in at most
$n$ transitions.

%
%
\setlength{\intextsep}{4pt}
\setlength{\textfloatsep}{10pt}
\begin{figure}[h]
\centering
\small
\parbox{0cm}{\begin{tabbing}
aa\=bb\=cc\= dd\= \kill
$\mi{MakeInv}(F,I,Q)$\{\\
\tb{\scriptsize{1}}\> while ($\mi{true}$) \}  \\
\tb{\scriptsize{2}}\Tt  $(\mi{Cex},Q) := \mi{MC}(F,I,Q)$\\
\tb{\scriptsize{3}}\Tt  if ($\mi{Cex} = \mi{nil}$) return($Q$) \\
\tb{\scriptsize{4}}\Tt  $Q := \mi{Relax}(Q,Cex)$\\
\tb{\scriptsize{5}}\Tt if ($Q \equiv 1$) return($Q$)\}\}\\
\end{tabbing}}
\vspace{-10pt}
\caption{The \ti{MakeInv} procedure}
\label{fig:make_inv}
\end{figure}

One can turn $Q$ into a property by using procedure \ti{MakeInv} shown
in Figure~\ref{fig:make_inv}. \ti{MakeInv} runs a 'while' loop.
First, \ti{MakeInv} calls a model checker \ti{MC}
(e.g. IC3~\cite{ic3}) to prove property $Q$. If \ti{MC} succeeds,
\ti{MakeInv} returns $Q$ as a property of $M$. Otherwise, \ti{MC}
finds a counterexample \ti{Cex}. This means that a state \pnt{s}
falsifying $Q$ (and thus unreachable in at most $n$ transitions) is
reachable in $i$ transitions where $i > n$. Then one needs to relax
$Q$ by replacing it with a property implied by $Q$ but not falsified
by \pnt{s}.

One way to relax $Q$ is to replace it with a solution $R$ to the PQE
problem of taking \prob{W}{Q(S) \wedge F(S,X,Y,S')} out of the scope
of quantifiers where $W = X \cup Y \cup S$. That is \prob{W}{Q \wedge
  F} $\equiv R\,\, \wedge$ \prob{W}{F}. Since the circuit $M$ has the
stuttering feature, $\prob{W}{F} \equiv 1$. So $R$ just specifies the
set of states reachable from $Q$-states in one transition. If \pnt{s}
still falsifies $R$, one can use PQE to find the set of states
reachable from $R$-states and so on. If \pnt{s} does not falsify $R$,
the latter is used as a new formula $Q$ (line 4). If relaxation ends
up with a trivial property, \ti{MakeInv} terminates (line
5). Otherwise a new iteration starts.

By taking different clauses of $F_{1,n}$ out of the scope of
quantifiers in \prob{W_{1,n}}{I_1 \wedge F_{1,n}} one can generate
different properties of the circuit $M$.

\section{Conclusions}
\label{sec:concl}
Incompleteness of a specification \ti{Spec} creates two problems.
First, an implementation \ti{Impl} of \ti{Spec} may have some
\ti{unwanted} properties that \ti{Spec} does not ban. Second,
\ti{Impl} may break some \ti{desired} properties that are not in
\ti{Spec}.  In either case, \ti{Spec} fails to expose bugs of
\ti{Impl}.  In testing, the problem of functional incompleteness is
addressed by running a test set that is complete \ti{structurally}
rather than functionally.  This structural completeness is achieved by
generating tests probing every piece of \ti{Impl}.  We apply this idea
to formal verification. Namely, we show that by using a technique
called partial quantifier elimination (PQE) one can generate
properties probing different parts of \ti{Impl}.  By checking that no
property of \ti{Impl} generated by PQE is unwanted one addresses the
first problem above. By updating \ti{Spec} to make it imply the
desired properties of \ti{Impl} generated by PQE one builds a
specification that is structurally complete. One can use a similar
approach to address the second problem above~\cite{inc_props}.

\bibliographystyle{plain}
\bibliography{short_sat,local}
\vspace{15pt}
\appendices
\section{Unwanted Property Derived By PQE}
\label{app:unwanted}
In this appendix, we give an example of an unwanted property derived
by PQE.  Consider the design of a combinational circuit called a
\ti{sorter}.  It accepts $r$-bit numbers ranging from 0 to
$2^r\!-\!1$, sorts them, and outputs the result.  Let $X$ and $Z$ be
sets of input and output variables of the sorter respectively. Assume
that the sorter accepts $m$ numbers. Let $\aps{x}_1,\dots,\aps{x}_m$
and $\aps{z}_1,\dots,\aps{z}_m$ be the numbers specified by input
\pnt{x} and output \pnt{z} respectively. The properties $P'(X)$ and
$P''(X,Z)$ below form a \ti{complete} specification of the sorter.
\begin{itemize}
\item $P'(\pnt{z}) = 1$ iff $\aps{z}_1 \le \dots \le \aps{z}_m$,
 \item $P''(\pnt{x},\pnt{z})\!=\!1$ iff
   $\aps{z}_1,\!\dots,\!\aps{z}_m$ is a permutation of
   $\aps{x}_1,\!\dots,\!\aps{x}_m$.
\end{itemize}
Let the designer use an \ti{incomplete} specification \abs{P}
consisting only of the property $P'$. Let $N(X,Y,Z)$ be an
\ti{implementation} of the sorter and $F(X,Y,Z)$ be a formula
describing the functionality of $N$. Assume $F \imp P'$ i.e. $N$
satisfies the specification \abs{P}.
Suppose $N$ is buggy. Namely, let $\aps{z}_1\!=\!0$ for every input
\pnt{x} of $N$.  (This does not contradict $F\!\imp\!P'$, since
$\aps{z}_i\!\ge\!0$, $1\!< i\!\le\!  m$.) Then $N$ has a property $Q$
falsified by the outputs \pnt{z} where $\aps{z}_1\!=\!b$ and $b$ is a
constant $1 \le b \le 2^r\!-\!1$.

Suppose $Q$ is obtained by taking $C \in F$ out of the scope of
quantifiers in \prob{Y}{F} i.e. by PQE.  On one hand, $P' \not\imp
Q$. Indeed, $P'$ is satisfied by an assignment \pnt{z} where
$\aps{z}_1,\dots,\aps{z}_m$ are sorted and $\aps{z}_1 = b$. So
derivation of $Q$ proves \abs{P} incomplete. On the other hand, $Q$ is
an \ti{unwanted} property of $N$. In a correct sorter, $\aps{z}_1$ can
take any value from 0 to $2^r\!-\!1$. So derivation of $Q$ exposes a
hole in \abs{P} and proves $N$ buggy.

\section{Proof Of Proposition~\ref{prop:lin_time}}
\label{app:proof}
\setcounter{proposition}{0}

%
%
\begin{proposition}
Let $\mi{Nlits}(F)$ denote the number of literals of $F$.  Let
\ti{QuickPQE} be applied to the PQE problem of taking $C'$ out of the
scope of quantifiers in \prob{Y}{C' \wedge F'} (described in
Subsection~\ref{ssec:pqe_prob}). Then \ti{QuickPQE} produces a correct
result and the complexity of \ti{QuickPQE} is $\abs{O}(\mi{Nlits}(F) +
|X|*|Z|)$.
\end{proposition}
%
%
\begin{proof}
The complexity of \ti{QuickPQE} is linear in $\mi{Nlits}(F)$ because
the former performs two test runs, each run having linear complexity
in the number of literals of $F$. The term $|X|*|Z|$ is due to the
fact that the solution produced by \ti{QuickPQE} may consist of $|Z|$
clauses of $|X|+1$ literals.

Now, let us show that \ti{QuickPQE} produces a correct solution. Let
$w \in Y \cup Z$ denote the output variable of the gate $g$. Assume
for the sake of clarity that $C$ contains the positive literal of
$w$. So, in the second run of \ti{QuickPQE} (where $C'$ and $C$ are
falsified) the value of $w$ is set to 0.

Denote the solution produced by \ti{QuickPQE} as $Q(X,Z)$ and so
$\prob{Y}{C' \wedge F'} \equiv Q \wedge \prob{Y}{F'}$.  We prove this
equivalence by showing that $C' \wedge F'$ and $Q \wedge F'$ are
equisatisfiable for every assignment (\pnt{x},\pnt{z}) to $X \cup
Z$. (Recall that $X$ and $Z$ specify the input and output variables of
the circuit $N$.) Below, we consider the three possible cases.

\tb{Case 1.} \ti{In the first run of QuickPQE, an \tb{input} variable
  of the gate} $g$ \ti{is assigned the value satisfying the clause}
$C$ \ti{(and hence the clause} $C'$\ti{)}. In this case $Q \equiv
1$. So one needs to show that for every assignment (\pnt{x},\pnt{z})
to $X \cup Z$, formulas $C' \wedge F'$ and $F'$ are
equisatisfiable. Consider the following two sub-cases.
\begin{enumerate}[(a)]
\item $\pnt{x} \neq \Pnt{x}$. Then $C'$ is satisfied by \pnt{x}
 and so $C' \wedge F'$ and $F'$ are logically equivalent in subspace
 (\pnt{x},\pnt{z}).
\item \pnt{x} = \Pnt{x}. In this case, $C'$ is satisfied by an
  assignment to an input variable of the gate $g$. The latter is true
  because the execution trace of $N$ under input \pnt{x} can be
  obtained by Boolean Constraint Operation (BCP) in subspace \pnt{x}
  over formula $C' \wedge F'$. The definition of an execution trace
  entails that this BCP satisfies \ti{all} clauses of $F$. The fact
  that BCP leads to satisfying $C'$ means that a clause implying $C'$
  in subspace \pnt{x} can be derived by
  resolving\footnote{Let clauses $C'$,$C''$ have opposite literals of exactly one variable
$w \in \V{C'} \cap \V{C''}$.  Then clauses $C'$,$C''$ are called
\ti{resolvable} on~$w$.  The clause $C$ having all literals of
$C',C''$ but those of $w$ is called the \ti{resolvent} of $C'$,$C''$
on $w$. The clause $C$ is said to be obtained by \ti{resolution} on
$w$.

} clauses of $F'$.  This means
  that $C'$ is implied by formula $F'$ in subspace \pnt{x}. So $C'
  \wedge F'$ and $F'$ are logically equivalent in subspace
  (\pnt{x},\pnt{z}).
\end{enumerate}

\tb{Case 2.} \ti{In the first run of QuickPQE, all \tb{input}
  variables of the gate} $g$ \ti{are assigned values \tb{falsifying}}
$C$. \ti{In the second run of QuickPQE}, $N$ \ti{outputs \tb{the same}
  assignment} \Pnt{z} \ti{as in the first run}. Then, like in the
first case, $Q \equiv 1$. So one needs to show that for every
assignment (\pnt{x},\pnt{z}) to $X \cup Z$, formulas $C' \wedge F'$
and $F'$ are equisatisfiable. Consider the following three sub-cases.
\begin{enumerate}[(a)]
\item $\pnt{x} \neq \Pnt{x}$. Then $C'$ is satisfied by \pnt{x}
 and so $C' \wedge F'$ and $F'$ are logically equivalent in subspace
 (\pnt{x},\pnt{z}).
\item $\pnt{x} = \Pnt{x}$ and $\pnt{z} \neq \Pnt{z}$. Let us show that
  in this case both $C' \wedge F'$ and $F'$ are unsatisfiable in
  subspace (\pnt{x},\pnt{z}).  Let $z_i \in Z$ be a variable assigned
  differently in \pnt{z} and \Pnt{z}. Let $l(z_i)$ be the literal of
  $z_i$ falsified by \pnt{z} (and satisfied by \Pnt{z}).  The fact
  that $N$ outputs \Pnt{z} under input \Pnt{x} means that $F$ implies
  the clause $\Cls{B}{x} \vee l(z_i)$.  So $C' \wedge F'$ is falsified
  in subspace (\pnt{x},\pnt{z}).  The fact that $N$ outputs \Pnt{z} in
  both runs means that $F'$ implies clauses $\Cls{B}{x} \vee l(z_i)
  \vee \overline{w}$ and $\Cls{B}{x} \vee l(z_i) \vee w$. (Recall that
  the variable $w$ specifies the output of the gate $g$. The variable
  $w$ is assigned 1 in the first run to satisfy $C'$ because the
  literals of all other variables of $C'$ are falsified. The variable
  $w$ is assigned 0 in the second run.) So $F'$ implies the resolvent
  of these two clauses on $w$ equal to $\Cls{B}{x} \vee l(z_i)$.
  Hence $F'$ is falsified in subspace (\pnt{x},\pnt{z}) too.

\item $\pnt{x} = \Pnt{x}$ and $\pnt{z} = \Pnt{z}$. Let us show that in
  this case $C' \wedge F'$ and $F'$ are both satisfiable in subspace
  (\pnt{x},\pnt{z}). Let \pnt{p} be the assignment to the variables of
  $N$ produced in the first run of \ti{QuickPQE}. By definition, the
  assignment to $X \cup Z$ in \pnt{p} is the same as in
  (\Pnt{x},\Pnt{z}) and hence in (\pnt{x},\pnt{z}). Besides, \pnt{p}
  satisfies $C' \wedge F'$ and hence $F'$.
\end{enumerate}

\tb{Case 3.} \ti{In the first run of QuickPQE, all \tb{input}
  variables of the gate} $g$ \ti{are assigned values \tb{falsifying}}
$C$. \ti{In the second run of QuickPQE}, $N$ \ti{outputs an
  assignment} \pnt{z^*} \ti{that is \tb{different} from the
  assignment} \Pnt{z} \ti{output in the second run of QuickPQE}. In
this case, the solution $Q(X,Z)$ consists of the clauses $\Cls{B}{x}
\vee l(z_1)$,\dots,$\Cls{B}{x} \vee l(z_p)$ where \s{z_1,\dots,z_p} is
the set of variables assigned differently in \Pnt{z} and \pnt{z^*}.
So one needs to show that for every assignment (\pnt{x},\pnt{z}) to $X
\cup Z$, formulas $C' \wedge F'$ and $Q \wedge F'$ are
equisatisfiable.  Consider the following four sub-cases. (We denote the
set of variables where \Pnt{z} and \pnt{z^*} have the same value as
$Z^*$.)
\begin{enumerate}[(a)]
\item $\pnt{x} \neq \Pnt{x}$. Then $C'$ and $Q$ are satisfied by
  \pnt{x}.  So $C' \wedge F'$ and $Q \wedge F'$ are logically
  equivalent in subspace (\pnt{x},\pnt{z}).
\item $\pnt{x} = \Pnt{x}$  and there is a variable
$z_i \in Z^*$ that is assigned in \pnt{z} differently than
in \Pnt{z}. Then both $C' \wedge F'$ and $Q \wedge F'$ are
unsatisfiable in subspace (\pnt{x},\pnt{z}). This can be shown as in
case 2b above.
\item $\pnt{x} = \Pnt{x}$ and all variables
of $Z^*$ are assigned the same value in \pnt{z} and \Pnt{z} and there
is a variable $z_i \in (Z \setminus Z^*)$ that is assigned in \pnt{z}
as in \pnt{z^*} (i.e. differently from \Pnt{z}). Let us show that in
this case both $C' \wedge F'$ and $Q \wedge F'$ are unsatisfiable in
subspace (\pnt{x},\pnt{z}). The formula $C' \wedge F'$ is falsified
because it implies the clause $\Cls{B}{x} \vee l(z_i)$ that is
falsified by (\pnt{x},\pnt{z}). The formula $Q \wedge F'$ is falsified
by (\pnt{x},\pnt{z}) because it contains the clause $\Cls{B}{x} \vee
l(z_i)$.
\item $\pnt{x} = \Pnt{x}$ and $\pnt{z} = \Pnt{z}$. Let us show that in
  this case $C' \wedge F'$ and $Q \wedge F'$ are both satisfiable in
  subspace (\pnt{x},\pnt{z}).  Let \pnt{p} be the assignment to the
  variables of $N$ produced in the first run of \ti{QuickPQE}. By
  definition, \pnt{p} agrees with assignment (\Pnt{x},\Pnt{z}) and
  hence with (\pnt{x},\pnt{z}). Besides, \pnt{p} satisfies $C' \wedge
  F'$ and hence $F'$. Since \pnt{p} also satisfies $Q$, it satisfies
  $Q \wedge F'$ as well.
\end{enumerate}
\end{proof}
%
%
\begin{proposition}
Let $C'$ be non-redundant in \prob{Y}{C' \wedge F'}. Then the formula
$Q(X,Z)$ generated by \ti{QuickPQE} is a single-test property of $N$.
\end{proposition}
%
%
\begin{proof}
Let $D$ be a clause of $Q(X,Z)$ i.e. $D = \Cls{B}{x} \vee l(z_i)$
where $z_i \in Z$ and \Cls{B}{x} is the longest clause falsified by
\Pnt{x}. The clause $D$ is satisfied by any assignment \pnt{x} to $X$
that is different from \Pnt{x}. Then $D$ and hence $Q$ meet the first
condition of Definition~\ref{def:st_prop}. Let \Pnt{z} denote the
output assignment produced by $N$ for \Pnt{x}. By definition of $Q$,
the literal $l(z_i)$ is satisfied by \Pnt{z}. So $D$ is satisfied by
(\Pnt{x},\Pnt{z}). Then $D$ and hence $Q$ meet the second condition of
Definition~\ref{def:st_prop}.  Finally, $D$ is falsified by the
assignment (\Pnt{x},\pnt{z^*}) where \pnt{z^*} is obtained from
\Pnt{z} by flipping the value of $z_i$.  Then $Q$ meets the third
condition of Definition~\ref{def:st_prop}, because it excludes the
output assignment \pnt{z^*} that is wrong for the input assignment
\pnt{x'}.
\end{proof}

\end{document}